\documentclass[journal]{IEEEtran}

\usepackage{amsmath,amsthm,amssymb,amsfonts}
\usepackage{bm}
\DeclareMathOperator*{\argmin}{arg\,min}

\theoremstyle{plain}
\newtheorem{theorem}{Theorem}

\newtheorem{corollary}{Corollary}

\theoremstyle{definition}
\newtheorem{definition}{Definition}
\newtheorem*{remark}{Remark}

\usepackage{balance}
\usepackage{graphicx}
\usepackage{cite}
\usepackage{microtype}
\usepackage[dvipsnames]{xcolor}
\usepackage{subcaption}
\usepackage{booktabs}

\usepackage{tikz}
\usetikzlibrary{
    arrows.meta,
    decorations.pathmorphing,
    patterns
}

\usepackage[
    colorlinks=true,
    linkcolor=blue,
    citecolor=blue,
    urlcolor=blue,
    bookmarks=true,
    pdfborder={0 0 0}
]{hyperref}
\usepackage{hypcap} 

\newcommand{\R}{\mathbb{R}}
\newcommand{\diff}{\mathrm{d}}

\title{From Distance to Angle: One-Shot Detection\\ Under Isotropic Multivariate Cauchy Noise}

\author{%
    Yen-Chi~Lee,~\IEEEmembership{Member,~IEEE}%
    \thanks{This work was supported by the National Science and Technology Council of Taiwan (NSTC 113-2115-M-008-013-MY3). (Corresponding author: Yen-Chi Lee.)}%
    \thanks{%
        Y.-C. Lee is with the Department of Mathematics, National Central University, Taoyuan, Taiwan (e-mail: \texttt{yclee@math.ncu.edu.tw}).%
    }%
}

\begin{document}
\maketitle

\begin{abstract}
We study one-shot detection under isotropic multivariate Cauchy noise using finite constellations, with emphasis on the geometric mechanisms governing symbol-level reliability. Under isotropic Cauchy noise, the maximum-likelihood rule induces the same Euclidean Voronoi decision regions as in the Gaussian case, so the distinction lies not in the decision geometry itself but in how probability mass is distributed over these fixed regions. 
In the small-noise regime, we derive a reciprocal distance-spectrum upper bound for the symbol error probability (SEP), showing that this bound, and the associated reliability descriptor, retain a longer-range dependence on the global constellation geometry than under additive white Gaussian noise. 
In the large-noise regime, we prove that the correct-decision probability converges to a limit determined solely by the angular measure of the associated Voronoi recession cone.
These results formalize a regime-dependent transition from bound-based distance descriptors to angle-based reliability descriptors under heavy-tailed noise.
Beyond asymptotic characterization, we show that these descriptors also admit a lightweight design interpretation for planar constellations under a common average power budget. The theory is further illustrated through an asymmetric four-point example exhibiting geometric collapse, a standard four-point Quadrature Amplitude Modulation (4QAM) sanity check, and finite-$\gamma$ numerical validation for both asymptotic regimes, together with descriptor-guided design comparisons that reveal collapse avoidance and reciprocal-distance burden as practically meaningful screening criteria.
\end{abstract}

\begin{IEEEkeywords}
Isotropic multivariate Cauchy noise, distance spectrum, heavy-tailed noise, one-shot detection, Voronoi cone.
\end{IEEEkeywords}

\section{Introduction}
\label{sec:intro}

\IEEEPARstart{A}{dditive} white Gaussian noise (AWGN) has long served as the canonical model for detection and constellation design, under which reliability is largely governed by minimum-distance geometry. This viewpoint underlies much of modern digital communication theory \cite{Proakis2008}, particularly in the small-noise regime. Recent studies, however, have shown that in a variety of communication and sensing settings, the disturbance is better modeled by heavy-tailed, non-Gaussian laws than by the Gaussian assumption \cite{samorodnitsky1994stable,nikias1995signal,georgiou2002alpha,ilow2002analytic}. These models have in turn motivated a substantial body of work across information theory, signal processing, and communication systems, including capacity analysis, robust estimation, and signal detection in impulsive environments \cite{fahs2012capacity,verdu2023cauchy,pang2024information,tsakalides1996robust,tsihrintzis2002performance,brown2000nonparametric}.

Specifically, the Cauchy model has emerged as a prototypical member of the $\alpha$-stable family \cite{nikias1995signal}, offering a mathematically rigorous yet tractable departure from the Gaussian assumption. While its infinite variance and polynomial tail decay complicate traditional concentration-based analysis (e.g., via the Law-of-Large-Numbers (LLN)) \cite{verdu2023cauchy}, the model is far from a mere theoretical abstraction. It arises naturally in practical systems, including \textit{massive MIMO} \cite{gulgun2023massive} and \textit{diffusion-based molecular communication} \cite{farsad2015stable,Lee_TCOM24}. At a broader level, existing work on stable and Cauchy noise has progressed along several largely separate directions. In information theory, the adopted input constraint often changes with the characteristic exponent: the Gaussian case is classically studied under a second-moment constraint, the Cauchy case has been analyzed under logarithmic cost \cite{fahs2014cauchy}, and more general $\alpha$-stable channels have also been treated under alternative moment constraints to obtain tractable capacity bounds \cite{defreitas2017capacity}. In signal processing, impulsive environments have motivated generalized information and estimation tools beyond the Gaussian second-order framework \cite{fahs2017information,tsakalides1996robust,tsihrintzis2002performance,brown2000nonparametric}. In communication systems, heavy-tailed laws also arise from concrete channel models rather than from abstract distributional choices alone; in particular, Cauchy-type behavior can naturally emerge in molecular communication settings such as first-arrival-position channels \cite{Lee_TCOM24,Lee2024ZeroDrift}.

These developments point to a common missing piece. Although recent findings \cite{pang2024information} show that memoryless scalar Cauchy-noise channels can sustain positive information rates under a second-moment input (energy) constraint, and thus support a renewed power-constrained viewpoint in the Cauchy regime, a multivariate geometric theory for finite-constellation detection under isotropic multivariate Cauchy noise remains undeveloped. In particular, from the perspective of symbol-level detection, it is still unclear which geometric descriptors replace the classical minimum-distance paradigm once the noise law becomes heavy-tailed.
This motivates a timely re-examination of \textit{finite-constellation detection and design} (cf. \cite{van2004,Proakis2008,porath2003design,yang2014design,luna2013constellation}) in the presence of Cauchy noise, especially if one wishes to extract design guidance that remains interpretable under heavy-tailed reliability laws.

To address this need, this paper develops a geometric asymptotic theory for one-shot detection over channels corrupted by isotropic multivariate Cauchy noise. Our main contributions are summarized as follows:
\begin{itemize}
    \item We formulate one-shot finite-constellation detection under \textit{isotropic multivariate Cauchy noise} and clarify that, although the maximum-likelihood rule induces the same Euclidean Voronoi decision regions as in the Gaussian case, the associated reliability laws are governed by fundamentally different geometric mechanisms.
    \item In the small-noise regime, we derive a reciprocal \textit{distance-spectrum} upper bound showing that the resulting one-shot reliability descriptor retains a longer-range sensitivity to the global constellation geometry than in the Gaussian setting.
    \item In the large-noise regime, we prove that bounded distance information disappears asymptotically and that the correct-decision probability converges to a limit determined solely by the angular measure of the associated \textit{Voronoi recession cone} (cf. \cite{okabe1992}), thereby identifying an angle-based reliability descriptor and a corresponding geometric-collapse criterion.
    \item We further support this regime-dependent transition through explicit worked examples and finite-$\gamma$ numerical validation, showing how the theory manifests in both asymmetric (see Appendix~\ref{app:four_point}) and symmetric (see Appendix~\ref{app:4qam}) constellations. In addition, for planar constellations under a common average power budget, we show through descriptor-guided design comparisons that the asymptotic quantities developed in the paper already possess nontrivial finite-$\gamma$ screening value, thereby providing a lightweight bridge from asymptotic analysis to Cauchy-aware constellation pre-design and candidate-geometry screening (cf. \cite{porath2003design,yang2014design,isaka2000multilevel}).
\end{itemize}

The remainder of this paper is organized as follows. Section~\ref{sec:model} introduces the system model and problem statement. Section~\ref{sec:small_noise} presents the small-noise distance-spectrum bound. Section~\ref{sec:large_noise} establishes the large-noise angular convergence theorem and the associated geometric-collapse criterion. Section~\ref{sec:design_implications} extracts minimal design implications for planar Cauchy constellations from the asymptotic descriptors. Section~\ref{sec:numerical_results} provides finite-$\gamma$ numerical validation for both asymptotic regimes, together with descriptor-guided design comparisons. The main results are further illustrated through an asymmetric four-point worked example in Appendix~\ref{app:four_point} and a standard $4$QAM sanity check in Appendix~\ref{app:4qam}. Section~\ref{sec:limitations_future} discusses limitations and future directions, and conclusions are drawn in Section~\ref{sec:conclusion}.

\section{System Model and Problem Statement}
\label{sec:model}

\subsection{Vector Cauchy Channels}
\label{sec:vector_cauchy}

We focus on one-shot detection (also referred to as symbol-by-symbol detection), where the receiver performs optimal maximum-likelihood (ML) decoding based solely on the current observation vector. This setting isolates the symbol-level effect of heavy-tailed noise on decision geometry and reliability, without invoking coding across multiple channel uses.

We consider a $d$-dimensional discrete-time vector channel model defined as \( \bm{y} = \bm{x} + \bm{n} \), where $\bm{x} \in \mathcal{C} \subset \R^d$ is the transmitted signal vector from a finite constellation subject to an average energy constraint $\mathbb{E}[\|\bm{x}\|^2] \le P$. The additive noise $\bm{n} \in \R^d$ is modeled as an \textit{isotropic} $d$-dimensional Cauchy random vector with probability density function (PDF) (cf. \cite{samorodnitsky1994stable,verdu2023cauchy}),
\begin{equation}
    f_{\bm{N}}(\bm{n}) = c_d\,\gamma^{-d} \left( 1+\frac{\|\bm{n}\|^2}{\gamma^2} \right)^{-\frac{d+1}{2}},\quad c_d=\frac{\Gamma\!\left(\frac{d+1}{2}\right)}{\pi^{\frac{d+1}{2}}},
\end{equation}
where $\gamma > 0$ is the scale parameter controlling the distribution's width. Here, the noise is modeled as a single isotropic multivariate Cauchy vector in $\R^d$, rather than as independent componentwise Cauchy coordinates. Unlike Gaussian noise, the Cauchy distribution lacks finite moments; thus, standard signal-to-noise ratio (SNR) is technically undefined. We instead use the normalized proxy $P/\gamma^2$ to characterize system performance (cf. \cite{pang2024information}).

In the Gaussian setting, isotropy and componentwise i.i.d. structure lead to the same multivariate law. In the Cauchy setting, however, these two notions are no longer equivalent: the isotropic multivariate Cauchy law is radially symmetric, whereas independent Cauchy coordinates define a different joint distribution. This distinction matters here because our analysis is carried out under an isotropic multivariate Cauchy observation model rather than a componentwise independent one. It also places the present work in contrast with the memoryless scalar Cauchy-noise setting studied in \cite{pang2024information}.

Under this spherically symmetric model, the ML decoding rule minimizes the negative log-likelihood, which reduces to the standard Euclidean distance metric. Namely,
\begin{equation}
    \hat{\bm{x}} = \argmin_{\bm{c}\in\mathcal{C}} \log\left(1 + \frac{\|\bm{y} - \bm{c}\|^2}{\gamma^2}\right) \equiv \argmin_{\bm{c}\in\mathcal{C}} \|\bm{y} - \bm{c}\|.
    \label{eq:ml_rule}
\end{equation}
Because the isotropic Cauchy density is monotonically decreasing with the Euclidean norm, the exact decision boundaries induced by \eqref{eq:ml_rule} correspond to Euclidean Voronoi cells (cf. \cite{okabe1992,fortune1986sweepline}), entirely independent of $\gamma$. 
Under equiprobable signaling, the average symbol error probability is formally given by
\begin{equation}
    P_e = 1 - \frac{1}{M} \sum_{\bm{x}\in\mathcal C} \int_{V(\bm{x})} f_{\bm{N}}(\bm{y}-\bm{x})~\diff\bm{y}.
    \label{eq:exact_pe}
\end{equation}

For the rest of this paper, we use uppercase letters such as $X,Y,N$ for random variables and lowercase letters such as $x,y,n$ for their realizations.

\subsection{The Combining Gain Paradox}
\label{sec:paradox}

In traditional AWGN vector channels, linear aggregation strategies (e.g., maximum ratio combining) often yield an averaging gain under standard finite-variance assumptions, so reliability can improve with the number of spatial dimensions or channel paths \cite{Proakis2008}. Under isotropic multivariate Cauchy noise, this self-averaging intuition breaks down: the Cauchy law has neither finite mean nor finite variance, so the usual LLN viewpoint no longer applies. More fundamentally, under the isotropic multivariate Cauchy model, every one-dimensional projection onto a unit direction remains a scalar Cauchy random variable with the same scale parameter (cf. \cite{samorodnitsky1994stable,nikias1995signal}). Thus, after normalizing a linear combiner to a unit direction, the projected noise remains Cauchy with scale parameter $\gamma$, rather than concentrating through averaging as in the Gaussian case. For a non-normalized linear functional, the scale changes only according to the norm of the combining vector, not through a variance-reduction mechanism. Consequently, linear combining does not reduce the effective noise scale in the usual Gaussian sense, but instead produces another scalar Cauchy projection, thereby removing the classical concentration effect associated with multidimensional combining.

Because the exact $d$-dimensional Cauchy integral in \eqref{eq:exact_pe} lacks a closed-form solution over arbitrary Voronoi polytopes, understanding this contrast calls for a viewpoint beyond linear averaging arguments alone. The specific objective of this paper is therefore to geometrically reinterpret the error integral \eqref{eq:exact_pe}. By showing how symbol reliability transitions from distance-based to angle-based behavior, we aim to identify the geometric descriptors that govern one-shot detection under the isotropic Cauchy model and that may inform future geometry-aware receiver design.

\section{Small-Noise Regime: Distance Spectrum Bounds}
\label{sec:small_noise}

In this section, we formally derive a distance-spectrum upper bound for the one-shot reliability in the small-noise regime.

\subsection{Pairwise Comparison Term}
\label{sec:pairwise_term}

Before deriving the small-noise bound, we isolate the basic pairwise comparison event that enters the union-bound argument. Fix two distinct constellation points $x_i,x_j\in\mathcal C$, and define
\begin{equation}
    d_{ij}:=\|x_i-x_j\|, \qquad u_{ij}:=\frac{x_j-x_i}{\|x_j-x_i\|}.
\end{equation}
Assume that $x_i$ is transmitted. Since the isotropic Cauchy law is absolutely continuous, tie events occur with probability zero and do not affect the ML analysis. Under ML decoding, the event that $x_j$ is at least as favorable as $x_i$ is \( \|Y-x_j\|\le \|Y-x_i\| \). Since $Y=x_i+N$, this is equivalent to
\begin{equation}
    \|N-(x_j-x_i)\|\le \|N\|,
\end{equation}
or, after expanding the squared norms, \( u_{ij}^{\mathsf T}N \ge \frac{d_{ij}}{2} \). Because the noise vector $N$ is isotropic Cauchy, every one-dimensional projection onto a unit direction remains a scalar Cauchy random variable with the same scale parameter $\gamma$. Therefore,
\begin{equation}
\begin{aligned}
    \mathbb P\!\left(\|Y-x_j\|\le \|Y-x_i\| \,\middle|\, X=x_i\right) &= \mathbb P\!\left(u_{ij}^{\mathsf T}N \ge \frac{d_{ij}}{2}\right) \\
    &= \frac{1}{2}-\frac{1}{\pi}\arctan\!\left(\frac{d_{ij}}{2\gamma}\right).
\end{aligned}
\label{eq:pairwise_halfspace_tail}
\end{equation}
Equation \eqref{eq:pairwise_halfspace_tail} shows that each pairwise comparison term depends only on the Euclidean distance $d_{ij}$. This quantity will serve as the basic building block in the distance-spectrum bound of Theorem~\ref{thm:small_noise}.

\subsection{Main Result: Small-Noise Bound and Asymptotics}
We define the distance spectrum of $\mathcal C$ as the multiset $\mathcal S(\mathcal C) = \{d_{ij}:1\le i<j\le M\}$. The following theorem establishes an asymptotic upper-bound description of the error probability as the noise scale vanishes.

\begin{theorem}[Small-Noise Distance-Spectrum Bound]
\label{thm:small_noise}
Let $\mathcal C=\{x_1,\dots,x_M\}\subset\mathbb R^d$ be a finite constellation, and consider the isotropic $d$-dimensional Cauchy channel \( Y=X+N \), with maximum-likelihood decoding as in \eqref{eq:ml_rule}. For a fixed transmitted symbol $x_i\in\mathcal C$, define
\begin{equation}
    d_{ij}:=\|x_i-x_j\|, \qquad j\neq i.
\end{equation}
Then, for every $\gamma>0$, the conditional symbol error probability satisfies
\begin{equation}
    P_e(x_i;\gamma) \le \sum_{j\neq i} \left[ \frac12-\frac{1}{\pi}\arctan\!\left(\frac{d_{ij}}{2\gamma}\right) \right].
    \label{eq:small_noise_union_bound}
\end{equation}
Consequently, as $\gamma\to 0$,
\begin{equation}
    P_e(x_i;\gamma) \le \frac{2\gamma}{\pi}\sum_{j\neq i}\frac{1}{d_{ij}} + O(\gamma^3).
    \label{eq:small_noise_asymp_symbol}
\end{equation}
If the symbols are equiprobable, then the average symbol error probability satisfies
\begin{equation}
    P_e(\gamma) \le \frac{4\gamma}{M\pi} \sum_{1\le i<j\le M}\frac{1}{d_{ij}} + O(\gamma^3).
    \label{eq:small_noise_asymp_average}
\end{equation}
Hence, in the small-noise regime, this union-bound descriptor for the one-shot reliability is controlled, at leading order, by the reciprocal distance spectrum of the constellation.
\end{theorem}

\begin{proof}
Fix $x_i\in\mathcal C$ and assume that $x_i$ is transmitted. Let \( E_i:=\{\hat X\neq x_i\} \) denote the decoding error event under ML decoding.

For each competitor $x_j\in\mathcal C\setminus\{x_i\}$, define
\begin{equation}
    H_{ij} := \Bigl\{ n\in\mathbb R^d: \|x_i+n-x_j\|\le \|n\| \Bigr\}.
\end{equation}
The event $H_{ij}$ means that, after adding noise $n$, the point $x_j$ is at least as close to the received vector as $x_i$. Therefore, \( E_i \subseteq \bigcup_{j\neq i} H_{ij} \), and hence, by the union bound,
\begin{equation}
    P_e(x_i;\gamma) \le \sum_{j\neq i}\mathbb P(H_{ij}).
    \label{eq:union_bound_step}
\end{equation}

We now simplify each event $H_{ij}$. Expanding the squared norms gives
\begin{equation}
\begin{aligned}
    &\|x_i+n-x_j\|^2 \le \|n\|^2 \\
    &\Longleftrightarrow\quad \|n-(x_j-x_i)\|^2 \le \|n\|^2,
\end{aligned}
\end{equation}
which is equivalent to
\begin{equation}
    -2(x_j-x_i)^{\mathsf T}n+\|x_j-x_i\|^2 \le 0.
\end{equation}
Thus,
\begin{equation}
    H_{ij} = \left\{ (x_j-x_i)^{\mathsf T}N \ge \frac{\|x_j-x_i\|^2}{2} \right\}.
\end{equation}
Let
\begin{equation}
    u_{ij}:=\frac{x_j-x_i}{\|x_j-x_i\|}=\frac{x_j-x_i}{d_{ij}}.
\end{equation}
Then
\begin{equation}
    H_{ij} = \left\{ u_{ij}^{\mathsf T}N \ge \frac{d_{ij}}{2} \right\}.
\end{equation}

Because $N$ is an isotropic multivariate Cauchy vector with scale parameter $\gamma$, every one-dimensional projection $u^{\mathsf T}N$ onto a unit vector $u$ is a one-dimensional Cauchy random variable with location $0$ and scale $\gamma$, i.e., \( u^{\mathsf T}N \sim \mathrm{Cauchy}(0,\gamma) \), whose density is
\begin{equation}
    f_Z(z)=\frac{1}{\pi\gamma}\frac{1}{1+(z/\gamma)^2}, \qquad z\in\mathbb R.
\end{equation}
Therefore,
\begin{equation}
\begin{aligned}
    \mathbb P(H_{ij}) &= \mathbb P\!\left(u_{ij}^{\mathsf T}N \ge \frac{d_{ij}}{2}\right) \\
    &= \int_{d_{ij}/2}^{\infty} \frac{1}{\pi\gamma}\frac{1}{1+(z/\gamma)^2}~\diff z.
\end{aligned}
\end{equation}
Evaluating the integral yields
\begin{equation}
    \mathbb P(H_{ij}) = \frac12-\frac{1}{\pi}\arctan\!\left(\frac{d_{ij}}{2\gamma}\right).
\end{equation}
Substituting this into \eqref{eq:union_bound_step} proves \eqref{eq:small_noise_union_bound}.

To obtain the small-noise expansion, use the standard asymptotic formula \cite{gradshteyn2007table}:
\begin{equation}
    \arctan t = \frac{\pi}{2}-\frac{1}{t}+\frac{1}{3t^3}+O(t^{-5}), \qquad t\to\infty.
\end{equation}
Setting $t=d_{ij}/(2\gamma)$ gives, for each fixed $i\neq j$,
\begin{equation}
    \frac12-\frac{1}{\pi}\arctan\!\left(\frac{d_{ij}}{2\gamma}\right) = \frac{2\gamma}{\pi d_{ij}}+O(\gamma^3), \qquad \gamma\to 0.
\end{equation}
Since the constellation is finite, summing over $j\neq i$ yields
\begin{equation}
    P_e(x_i;\gamma) \le \frac{2\gamma}{\pi}\sum_{j\neq i}\frac{1}{d_{ij}} + O(\gamma^3),
\end{equation}
which proves \eqref{eq:small_noise_asymp_symbol}.

Finally, if the input symbols are equiprobable, then
\begin{equation}
    P_e(\gamma) = \frac{1}{M}\sum_{i=1}^M P_e(x_i;\gamma).
\end{equation}
Averaging the previous bound over $i$ gives
\begin{equation}
    P_e(\gamma) \le \frac{2\gamma}{M\pi} \sum_{i=1}^M\sum_{j\neq i}\frac{1}{d_{ij}} + O(\gamma^3).
\end{equation}
Using the symmetry $d_{ij}=d_{ji}$ and double counting,
\begin{equation}
    \sum_{i=1}^M\sum_{j\neq i}\frac{1}{d_{ij}} = 2\sum_{1\le i<j\le M}\frac{1}{d_{ij}},
\end{equation}
which proves \eqref{eq:small_noise_asymp_average}.
\end{proof}

\begin{remark}
Theorem~\ref{thm:small_noise} clarifies an important contrast between Gaussian and Cauchy detection. 
Under AWGN, pairwise error terms decay exponentially with distance, so the nearest-neighbor structure typically dominates the small-noise behavior \cite{simon2002some}. 
Under the isotropic Cauchy model, however, the pairwise comparison term in \eqref{eq:pairwise_halfspace_tail} satisfies
\begin{equation}
    \frac{1}{2}-\frac{1}{\pi}\arctan\!\left(\frac{d_{ij}}{2\gamma}\right) = \frac{2\gamma}{\pi d_{ij}}+O(\gamma^3), \qquad \gamma \to 0,
\end{equation}
which decays only algebraically in the distance $d_{ij}$. As a result, farther constellation points are suppressed more weakly than in the Gaussian case, and the small-noise error bound naturally depends on the full reciprocal distance spectrum rather than solely on the minimum distance \cite{martinez2005error}.
In this sense, Theorem~\ref{thm:small_noise} should be interpreted as a geometric upper bound showing that, even in the high-reliability regime, the Cauchy pairwise-comparison and union-bound analysis preserves a longer-range sensitivity to the global constellation geometry.
\end{remark}

\section{Large-Noise Regime: Angular Voronoi Convergence}
\label{sec:large_noise}

As $\gamma$ increases, the distance-based description breaks down. In this regime, the received vector is pushed far away from the transmitted symbol, so the bounded offsets among constellation points become negligible relative to the overall noise magnitude. We show that the decoding problem then loses sensitivity to finite Euclidean distances and retains only the asymptotic angular geometry of the Voronoi cells. The appropriate descriptor is therefore the recession cone of each Voronoi region and its associated spherical angular patch.

\subsection{Voronoi Cones and Angular Patches}
For each $x\in\mathcal C$, the associated Voronoi cell (cf. \cite{okabe1992,allevi2024basic}) is
\begin{equation}
    V(x)=\{y\in\mathbb R^d:\ \|y-x\|\le \|y-c\|,\ \forall c\in\mathcal C\}.
\end{equation}
Its recession cone is defined by
\begin{equation}
    \mathsf K(x) = \bigcap_{c\in\mathcal C\setminus\{x\}} \{u\in\mathbb R^d:\ u^{\mathsf T}(c-x)\le 0\}.
    \label{eq:cone_def}
\end{equation}
The spherical cross-section
\(
    A_x:=\mathsf K(x)\cap\mathbb S^{d-1}
\)
is called the \emph{Voronoi angular patch} of $x$, where $\mathbb S^{d-1}$ denotes the unit sphere in $\mathbb R^d$.

The geometric meaning of \eqref{eq:cone_def} is transparent after translating the Voronoi cell to the origin. Indeed,
\begin{equation}
    V(x)-x = \bigcap_{c\in\mathcal C\setminus\{x\}} \left\{ z\in\mathbb R^d:\ z^{\mathsf T}(c-x)\le \frac{\|c-x\|^2}{2} \right\}.
\end{equation}
Hence, after scaling by $\gamma$,
\begin{equation}
    \frac{V(x)-x}{\gamma} = \bigcap_{c\in\mathcal C\setminus\{x\}} \left\{ u\in\mathbb R^d:\ u^{\mathsf T}(c-x)\le \frac{\|c-x\|^2}{2\gamma} \right\},
    \label{eq:scaled_cell}
\end{equation}
which suggests the recession-cone limit \eqref{eq:cone_def} as $\gamma\to\infty$. Thus, the Voronoi cone $\mathsf K(x)$ is precisely the large-noise limit of the translated-and-rescaled decision region.

\subsection{Main Result: Large-Noise Convergence}
The following theorem shows that, in the infinite-noise limit, the correct-decision probability depends only on the angular measure of the recession cone.

\begin{theorem}[Large-Noise Convergence]
\label{thm:large_noise}
Let $\mathcal C\subset\mathbb R^d$ be a finite constellation, and let
\begin{equation}
    P_c(x;\gamma):=\mathbb P(\hat X=x\,|\,X=x)
\end{equation}
denote the conditional correct-decision probability under the isotropic $d$-dimensional Cauchy channel with scale parameter $\gamma$. For each $x\in\mathcal C$, let $\mathsf K(x)$ and $A_x$ be defined by \eqref{eq:cone_def} and the spherical cross-section, and let $\sigma_{d-1}$ denote the usual surface measure on $\mathbb S^{d-1}$. Then
\begin{equation}
    \lim_{\gamma\to\infty} P_c(x;\gamma) = \frac{\sigma_{d-1}(A_x)}{\sigma_{d-1}(\mathbb S^{d-1})}.
    \label{eq:large_noise_limit}
\end{equation}
Equivalently,
\begin{equation}
    \lim_{\gamma\to\infty} P_e(x;\gamma) = 1-\frac{\sigma_{d-1}(A_x)}{\sigma_{d-1}(\mathbb S^{d-1})},
    \label{eq:large_noise_error_limit}
\end{equation}
where $P_e(x;\gamma)=1-P_c(x;\gamma)$.

If the constellation symbols are equiprobable, define
\begin{equation}
    P_c(\gamma):=\frac{1}{M}\sum_{x\in\mathcal C} P_c(x;\gamma).
\end{equation}
Then
\begin{equation}
    \lim_{\gamma\to\infty} P_c(\gamma) = \frac{1}{M}\sum_{x\in\mathcal C} \frac{\sigma_{d-1}(A_x)}{\sigma_{d-1}(\mathbb S^{d-1})}.
    \label{eq:large_noise_average_limit}
\end{equation}
\end{theorem}

\begin{proof}
Fix $x\in\mathcal C$. By definition of ML decoding,
\begin{equation}
    P_c(x;\gamma) = \mathbb P(Y\in V(x)\mid X=x) = \mathbb P(N\in V(x)-x).
\end{equation}
Since the noise density is
\begin{equation}
    f_N(n) = c_d\,\gamma^{-d}\left(1+\frac{\|n\|^2}{\gamma^2}\right)^{-\frac{d+1}{2}},
\end{equation}
we obtain
\begin{equation}
    P_c(x;\gamma) = \int_{V(x)-x} c_d\,\gamma^{-d}\left(1+\frac{\|n\|^2}{\gamma^2}\right)^{-\frac{d+1}{2}}~\diff n.
\end{equation}
Applying the change of variables $n=\gamma u$ yields
\begin{equation}
    P_c(x;\gamma) = \int_{(V(x)-x)/\gamma} c_d\,(1+\|u\|^2)^{-\frac{d+1}{2}}~\diff u.
    \label{eq:pc_scaled_integral}
\end{equation}

We now show that the indicator of the scaled region converges almost everywhere to the indicator of the recession cone. From \eqref{eq:scaled_cell},
\begin{equation}
    \frac{V(x)-x}{\gamma} = \bigcap_{c\in\mathcal C\setminus\{x\}} \left\{ u:\ u^{\mathsf T}(c-x)\le \frac{\|c-x\|^2}{2\gamma} \right\}.
\end{equation}
Let
\begin{equation}
    B_x := \bigcup_{c\in\mathcal C\setminus\{x\}} \{u\in\mathbb R^d:\ u^{\mathsf T}(c-x)=0\}.
\end{equation}
Since $\mathcal C$ is finite, $B_x$ is a finite union of hyperplanes and therefore has Lebesgue measure zero.

Now fix $u\notin B_x$. Then for every $c\neq x$, the quantity $u^{\mathsf T}(c-x)$ is nonzero. If $u\in\mathsf K(x)$, then all these inner products are strictly negative, so for sufficiently large $\gamma$ they satisfy
\begin{equation}
    u^{\mathsf T}(c-x)\le \frac{\|c-x\|^2}{2\gamma}, \qquad \forall c\neq x,
\end{equation}
and hence $u\in (V(x)-x)/\gamma$. Conversely, if $u\notin\mathsf K(x)$, then there exists some $c\neq x$ such that $u^{\mathsf T}(c-x)>0$, and for sufficiently large $\gamma$,
\begin{equation}
    u^{\mathsf T}(c-x)>\frac{\|c-x\|^2}{2\gamma},
\end{equation}
so $u\notin (V(x)-x)/\gamma$. Therefore,
\begin{equation}
    \mathbf 1_{(V(x)-x)/\gamma}(u)\to \mathbf 1_{\mathsf K(x)}(u) \qquad\text{for a.e. }u\in\mathbb R^d.
    \label{eq:indicator_convergence}
\end{equation}

Since
\begin{equation}
    0 \le \mathbf 1_{(V(x)-x)/\gamma}(u)\, c_d(1+\|u\|^2)^{-\frac{d+1}{2}} \le c_d(1+\|u\|^2)^{-\frac{d+1}{2}},
\end{equation}
and the dominating function on the right is integrable over $\mathbb R^d$, the dominated convergence theorem applied to \eqref{eq:pc_scaled_integral} and \eqref{eq:indicator_convergence} gives
\begin{equation}
    \lim_{\gamma\to\infty}P_c(x;\gamma) = \int_{\mathsf K(x)} c_d(1+\|u\|^2)^{-\frac{d+1}{2}}~\diff u.
    \label{eq:cone_integral_limit}
\end{equation}

It remains to evaluate the right-hand side in angular form. Because the density is radial, we use polar coordinates $u=r\theta$, where $r\in[0,\infty)$ and $\theta\in\mathbb S^{d-1}$. Since $\mathsf K(x)$ is a cone, its polar decomposition is exactly \( \mathsf K(x)=\{r\theta:\ r\ge 0,\ \theta\in A_x\} \). Hence
\begin{equation}
\begin{aligned}
    &\int_{\mathsf K(x)} c_d(1+\|u\|^2)^{-\frac{d+1}{2}} \diff u \\
    &= \int_{A_x} \! \int_0^\infty c_d(1+r^2)^{-\frac{d+1}{2}} r^{d-1} \diff r \diff \sigma_{d-1}(\theta) \\
    &= \sigma_{d-1}(A_x) \underbrace{\int_0^\infty c_d(1+r^2)^{-\frac{d+1}{2}} r^{d-1} \diff r}_{=:I_d}.
\end{aligned}
\label{eq:polar_cone}
\end{equation}
On the other hand, integrating the same radial density over the entire space $\mathbb R^d$ gives
\begin{equation}
    1 = \int_{\mathbb R^d} c_d(1+\|u\|^2)^{-\frac{d+1}{2}}~\diff u = \sigma_{d-1}(\mathbb S^{d-1})\, I_d.
\end{equation}
Thus
\begin{equation}
    I_d=\frac{1}{\sigma_{d-1}(\mathbb S^{d-1})}.
\end{equation}
Substituting this into \eqref{eq:polar_cone} and then into \eqref{eq:cone_integral_limit} yields
\begin{equation}
    \lim_{\gamma\to\infty}P_c(x;\gamma) = \frac{\sigma_{d-1}(A_x)}{\sigma_{d-1}(\mathbb S^{d-1})},
\end{equation}
which proves \eqref{eq:large_noise_limit}. Equation \eqref{eq:large_noise_error_limit} follows from $P_e(x;\gamma)=1-P_c(x;\gamma)$, and \eqref{eq:large_noise_average_limit} follows by averaging over equiprobable symbols.
\end{proof}

\begin{corollary}[Condition for Geometric Collapse]
\label{cor:collapse}
Let $x\in\mathcal C$. If the angular patch $A_x$ has zero surface measure, i.e.,
\begin{equation}
    \sigma_{d-1}(A_x)=0,
\end{equation}
then
\begin{equation}
    \lim_{\gamma\to\infty}P_c(x;\gamma)=0, \qquad \lim_{\gamma\to\infty}P_e(x;\gamma)=1.
    \label{eq:collapse_limit}
\end{equation}
In particular, if $x$ lies in the interior of $\operatorname{conv}(\mathcal C)$, then $A_x=\varnothing$, and therefore geometric collapse occurs.
\end{corollary}

\begin{proof}
If $\sigma_{d-1}(A_x)=0$, then \eqref{eq:large_noise_limit} in Theorem~\ref{thm:large_noise} immediately gives
\begin{equation}
    \lim_{\gamma\to\infty}P_c(x;\gamma) = \frac{\sigma_{d-1}(A_x)}{\sigma_{d-1}(\mathbb S^{d-1})} = 0,
\end{equation}
which proves \eqref{eq:collapse_limit}.

For the interior-point claim, suppose that $x\in \operatorname{int}(\operatorname{conv}(\mathcal C))$. If there existed some nonzero $u\in\mathsf K(x)$, then by definition of $\mathsf K(x)$,
\begin{equation}
    u^{\mathsf T}(c-x)\le 0, \qquad \forall c\in\mathcal C.
\end{equation}
By convexity, this would imply
\begin{equation}
    u^{\mathsf T}(y-x)\le 0, \qquad \forall y\in\operatorname{conv}(\mathcal C).
\end{equation}
However, since $x$ is an interior point of $\operatorname{conv}(\mathcal C)$, there exists $\varepsilon>0$ such that
\begin{equation}
    x+\varepsilon \frac{u}{\|u\|}\in \operatorname{conv}(\mathcal C),
\end{equation}
which would give
\begin{equation}
    u^{\mathsf T}\!\left(x+\varepsilon \frac{u}{\|u\|}-x\right) = \varepsilon \|u\|>0,
\end{equation}
a contradiction. Therefore $\mathsf K(x)$ contains no nonzero direction, so
\begin{equation}
    A_x=\mathsf K(x)\cap\mathbb S^{d-1}=\varnothing.
\end{equation}
This completes the proof.
\end{proof}


\section{Design Implications for 2D Cauchy Constellations}
\label{sec:design_implications}

The results developed in Sections~\ref{sec:small_noise} and~\ref{sec:large_noise} are analytical in nature, but they already suggest nontrivial design implications for planar constellations under isotropic multivariate Cauchy noise. To keep the present extension lightweight, we restrict attention here to two-dimensional, equiprobable, finite constellations under a common average power budget. Our aim is not to formulate or solve a full constrained optimization problem, but rather to extract simple, design-oriented descriptor-level principles that may guide future Cauchy-aware constellation design, screening, and refinement.

\subsection{Hull-Only Principle and Collapse Avoidance}
\label{sec:hull_only_principle}

We begin with the large-noise regime. For a point $x_i\in\mathcal C\subset\mathbb R^2$, let \( A_{x_i}:=\mathsf K(x_i)\cap \mathbb S^1 \) denote its Voronoi angular patch, as defined in Section~\ref{sec:large_noise}. To avoid confusion with this patch notation, we define the corresponding \emph{normalized angular robustness} by
\begin{equation}
    A_i(\mathcal C) := \frac{\sigma_1(A_{x_i})}{\sigma_1(\mathbb S^1)},
    \label{eq:angular_robustness}
\end{equation}
where $\sigma_1$ is the arc-length measure on $\mathbb S^1$. By Theorem~\ref{thm:large_noise},
\begin{equation}
    \lim_{\gamma\to\infty} P_c(x_i;\gamma)=A_i(\mathcal C).
    \label{eq:angular_robustness_limit}
\end{equation}

Corollary~\ref{cor:collapse} already shows that if $x_i\in \operatorname{int}(\operatorname{conv}(\mathcal C))$, then $A_i(\mathcal C)=0$, so geometric collapse occurs. In the planar setting, this observation can be sharpened geometrically: a positive angular patch can arise only at a vertex of the convex hull. Indeed, if a point lies in the interior of the hull, or on the relative interior of one of its edges, then the associated recession cone has zero angular width, hence zero arc-length on $\mathbb S^1$. By contrast, if $x_i$ is a hull vertex, then $A_i(\mathcal C)$ is exactly the corresponding exterior-angle fraction. That is,
\begin{equation}
    A_i(\mathcal C)=\frac{\phi_i}{2\pi},
    \label{eq:exterior_angle_fraction}
\end{equation}
where $\phi_i$ denotes the exterior angle of $\operatorname{conv}(\mathcal C)$ at $x_i$.

This immediately suggests a simple but important design principle: if one wishes every symbol to retain a nonzero large-noise correct-decision limit, then one should avoid interior placements and favor \emph{hull-only} constellations. Moreover, even among hull vertices, extremely acute geometric configurations produce very small exterior angles and hence very small $A_i(\mathcal C)$. Therefore, from the large-noise viewpoint, a robust planar Cauchy constellation should not only avoid interior points, but also avoid highly unbalanced hull geometries that leave some symbols with vanishingly small angular protection. In this sense, $A_i(\mathcal C)$ serves as a symbol-level asymptotic robustness score, while $A_{\min}(\mathcal C):=\min_i A_i(\mathcal C)$ provides a natural worst-symbol screening metric for collapse avoidance.

\subsection{Reciprocal-Distance Burden}
\label{sec:reciprocal_distance_burden}

We next turn to the small-noise regime. Theorem~\ref{thm:small_noise} shows that, for each transmitted symbol $x_i$,
\begin{equation}
    P_e(x_i;\gamma) \le \frac{2\gamma}{\pi}\sum_{j\neq i}\frac{1}{\|x_i-x_j\|} + O(\gamma^3), \qquad \gamma\to 0.
    \label{eq:small_noise_burden_bound}
\end{equation}
This motivates the following symbol-level descriptor.

\begin{definition}
For a finite constellation $\mathcal C=\{x_1,\dots,x_M\}\subset\mathbb R^2$, define the \emph{reciprocal-distance burden} of symbol $x_i$ by
\begin{equation}
    B_i(\mathcal C) := \sum_{j\neq i}\frac{1}{\|x_i-x_j\|}.
    \label{eq:burden_symbol}
\end{equation}
We also define the corresponding worst-symbol burden by
\begin{equation}
    B_{\max}(\mathcal C) := \max_{1\le i\le M} B_i(\mathcal C).
    \label{eq:burden_max}
\end{equation}
\end{definition}

Under this notation, \eqref{eq:small_noise_burden_bound} becomes
\begin{equation}
    P_e(x_i;\gamma) \le \frac{2\gamma}{\pi} B_i(\mathcal C)+O(\gamma^3), \qquad \gamma\to 0.
    \label{eq:burden_rewrite}
\end{equation}
The quantity $B_i(\mathcal C)$ should be viewed as a Cauchy-specific measure of how ``crowded'' symbol $x_i$ is from the small-noise viewpoint. Unlike the Gaussian nearest-neighbor heuristic, $B_i(\mathcal C)$ aggregates the influence of \emph{all} other constellation points, reflecting the slow algebraic suppression of distant competitors under Cauchy tails.

This interpretation leads to a second design implication. Under a common power budget, a planar Cauchy constellation should not be judged solely by its minimum distance. It should also avoid geometric arrangements that create a large reciprocal-distance burden for some symbols, such as severe clustering, strongly uneven point density, or configurations in which one point is surrounded too closely by many others. In this sense, $B_{\max}(\mathcal C)$ provides a conservative bound-based design-screening descriptor for the high-reliability regime: a smaller value indicates weaker upper-bound worst-symbol sensitivity to the global constellation geometry and hence suggests a more balanced small-noise reliability profile.

\subsection{A Minimal Joint Design Objective}
\label{sec:minimal_joint_design_objective}

The descriptors above correspond to two opposite asymptotic regimes. The quantity $B_{\max}(\mathcal C)$ captures the bound-based worst-symbol vulnerability in the small-noise regime, whereas the quantity
\begin{equation}
    A_{\min}(\mathcal C) := \min_{1\le i\le M} A_i(\mathcal C)
    \label{eq:angular_robustness_min}
\end{equation}
captures the weakest large-noise angular robustness across symbols. Taken together, they suggest a minimal joint design viewpoint.

Specifically, for planar constellations normalized to a common average power budget $P_0$, one may compare candidate designs using a power-normalized surrogate criterion. To keep the two terms dimensionally compatible, define the normalized worst-symbol burden
\begin{equation}
    \overline B_{\max}(\mathcal C) := \sqrt{P_0}\,B_{\max}(\mathcal C).
    \label{eq:normalized_burden_max}
\end{equation}
Since $P_0$ is fixed across the candidate class, this normalization does not change the ordering induced by $B_{\max}(\mathcal C)$. We then use the surrogate criterion
\begin{equation}
    J_{\lambda}(\mathcal C) := \lambda\, \overline B_{\max}(\mathcal C) - (1-\lambda)\, A_{\min}(\mathcal C), \qquad 0\le \lambda\le 1.
    \label{eq:minimal_joint_objective}
\end{equation}
Here $\lambda$ controls the relative emphasis placed on the small-noise and large-noise descriptors. A larger $\lambda$ favors constellations with smaller power-normalized reciprocal-distance burden, whereas a smaller $\lambda$ places greater emphasis on avoiding angular fragility and geometric collapse.

We stress that \eqref{eq:minimal_joint_objective} is \emph{not} proposed as an exact finite-$\gamma$ performance metric, nor as a complete optimization formulation. Rather, it serves as a lightweight descriptor-level surrogate suggested directly by Theorems~\ref{thm:small_noise} and~\ref{thm:large_noise}. Its role is mainly comparative: it provides a simple way to rank normalized candidate constellations according to whether they simultaneously (i) avoid collapse in the large-noise regime and (ii) reduce worst-symbol reciprocal-distance crowding in the small-noise regime. From a signal-design viewpoint, \eqref{eq:minimal_joint_objective} may therefore be interpreted as a low-complexity pre-design proxy for screening candidate geometries before more detailed finite-$\gamma$ evaluation.

Accordingly, the present theory already suggests a basic two-stage design philosophy for planar Cauchy constellations. First, discard geometries with $A_{\min}(\mathcal C)=0$, since these necessarily contain symbols that collapse asymptotically. Second, among the remaining normalized candidates, prefer those with smaller $\overline B_{\max}(\mathcal C)$, equivalently smaller $B_{\max}(\mathcal C)$ within a fixed-power candidate class, while monitoring how much angular imbalance is introduced. This viewpoint does not yet constitute a full synthesis theory, but it shows that the asymptotic analysis developed in this paper already yields concrete and nontrivial design guidance beyond pure performance characterization, and already points toward a descriptor-driven methodology for Cauchy-aware constellation pre-design in two dimensions.


\section{Numerical Results}
\label{sec:numerical_results}

\begin{figure*}[!t]
    \centering
    \begin{subfigure}[t]{0.4\textwidth}
        \centering
        \includegraphics[width=\linewidth]{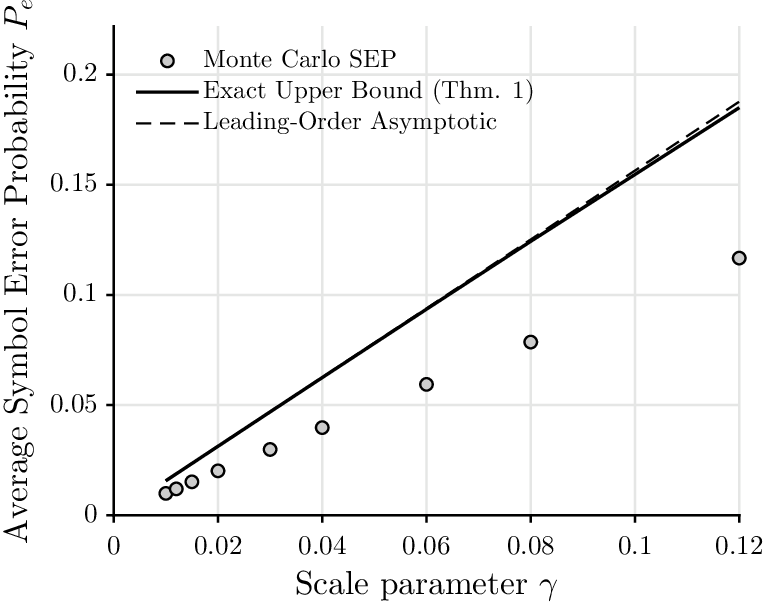}
        \caption{Small-noise average symbol error probability.}
        \label{fig:sim_small_gamma_avg}
    \end{subfigure}
    \hfill
    \begin{subfigure}[t]{0.4\textwidth}
        \centering
        \includegraphics[width=\linewidth]{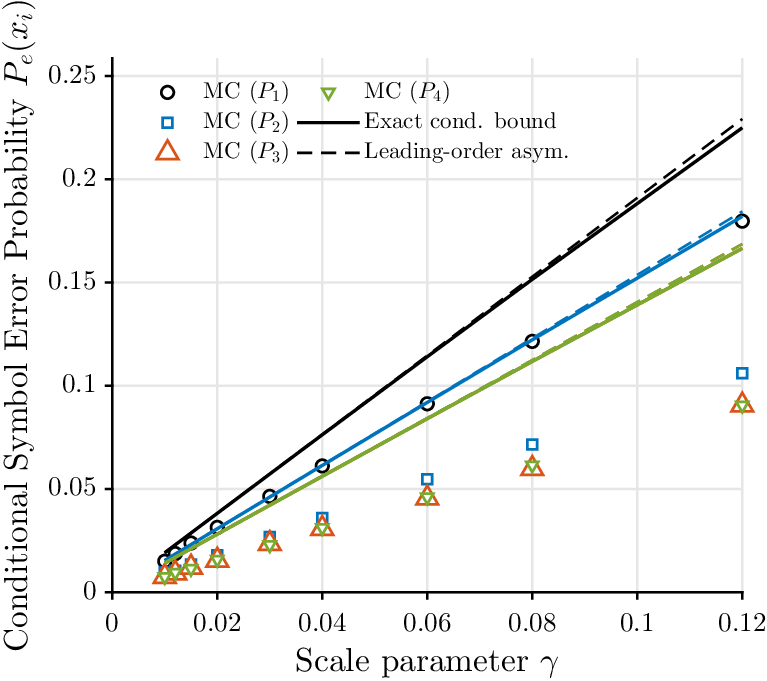}
        \caption{Small-noise conditional symbol error probabilities.}
        \label{fig:sim_small_gamma_cond}
    \end{subfigure}

    \vspace{0.35em}

    \begin{subfigure}[t]{0.4\textwidth}
        \centering
        \includegraphics[width=\linewidth]{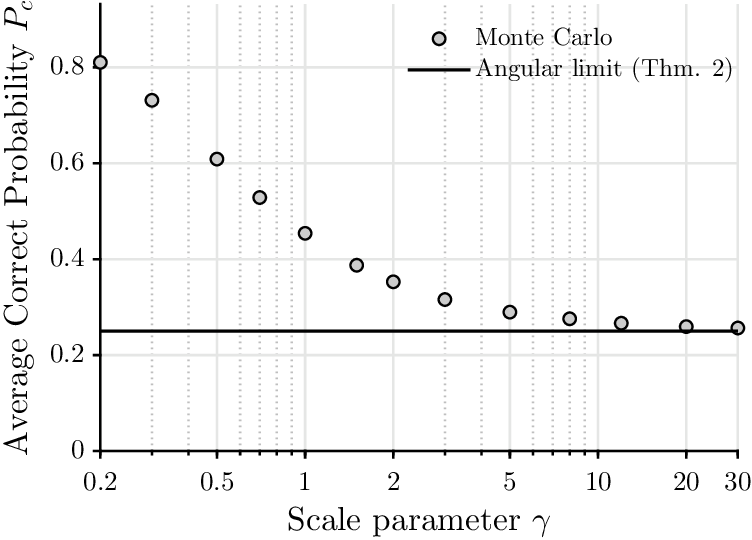}
        \caption{Large-noise average correct-decision probability.}
        \label{fig:sim_large_gamma_avg}
    \end{subfigure}
    \hfill
    \begin{subfigure}[t]{0.4\textwidth}
        \centering
        \includegraphics[width=\linewidth]{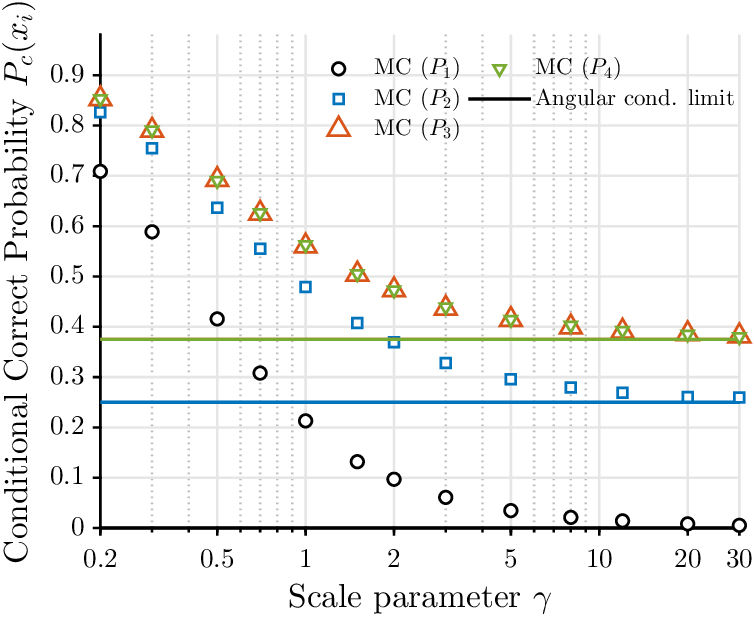}
        \caption{Large-noise conditional correct-decision probabilities.}
        \label{fig:sim_large_gamma_cond}
    \end{subfigure}
    \caption{Finite-$\gamma$ validation of the baseline asymptotic descriptors for the asymmetric four-point constellation in Appendix~\ref{app:four_point}.
    \textbf{Panels (a)--(b)} illustrate the small-noise upper-bound descriptor in Theorem~\ref{thm:small_noise}, while \textbf{panels (c)--(d)} illustrate convergence toward the angular limits in Theorem~\ref{thm:large_noise}, including the collapse of $P_1$.}
    \label{fig:baseline_validations}
\end{figure*}

Unless otherwise stated, all Monte Carlo (MC) experiments in this section use equiprobable signaling, a fixed random seed \texttt{rng(2026,'twister')}, and \(N_{\mathrm{mc}}=5\times 10^5\) trials for each value of \(\gamma\), processed in batches of size \(10^5\). The additive noise is generated as an isotropic multivariate Cauchy vector in \(\mathbb{R}^2\) via
\begin{equation}
    N=\gamma\,\frac{G}{H},
    \label{eq:sim_common_sampler}
\end{equation}
where \(G\sim\mathcal{N}(0,I_2)\) and \(H\sim\mathcal{N}(0,1)\) are independent. In all simulations, maximum-likelihood decoding is implemented by the equivalent Euclidean nearest-neighbor rule in \eqref{eq:ml_rule}. For each experiment, the reported average and symbol-level quantities are computed from the same underlying Monte Carlo runs under the experiment-specific \(\gamma\)-grid stated explicitly below. The numerical study has two roles: first, to confirm that the asymptotic descriptors developed in Sections~\ref{sec:small_noise} and~\ref{sec:large_noise} are already visible at finite noise scales; and second, to assess whether these descriptors provide useful guidance for the design comparisons in Section~\ref{sec:design_implications}.

\subsection{Finite-\texorpdfstring{$\gamma$}{gamma} Validation of Theorem~\ref{thm:small_noise}}
\label{sec:sim_small_gamma}

We validate Theorem~\ref{thm:small_noise} in the finite small-noise regime using the asymmetric four-point constellation from Appendix~\ref{app:four_point},
\begin{equation}
    \mathcal{X}=\{(0,0),(1,0),(0,1),(0,-1)\}.
    \label{eq:sim1_constellation}
\end{equation}
For each
\begin{equation}
    \gamma \in \{0.01, 0.012, 0.015, 0.02, 0.03, 0.04, 0.06, 0.08, 0.12\},
    \label{eq:sim1_gamma_grid}
\end{equation}
we estimate the symbol error probability by Monte Carlo simulation and compare it with both the exact union bound from Theorem~\ref{thm:small_noise} and its leading-order small-$\gamma$ asymptotic expression.

Fig.~\ref{fig:baseline_validations}(a)--(b) shows that the small-noise descriptor from Theorem~\ref{thm:small_noise} is already visible at finite noise scales. In Fig.~\ref{fig:baseline_validations}(a), the Monte Carlo average symbol error probability remains below the exact upper bound over the tested range, while the exact bound and its leading-order asymptotic approximation are nearly indistinguishable. Thus, although the upper bound is not numerically tight, it correctly captures the relevant small-noise scaling.

Fig.~\ref{fig:baseline_validations}(b) further resolves the symbol-level behavior. The observed ordering
\begin{equation}
    P_e(P_1;\gamma) > P_e(P_2;\gamma) > P_e(P_3;\gamma) = P_e(P_4;\gamma)
    \label{eq:sim1_ordering}
\end{equation}
is clearly visible and cannot be explained by \(d_{\min}\) alone, since all four points share the same minimum distance. Instead, it is consistent with the bound-based reciprocal-distance coefficients in Theorem~\ref{thm:small_noise}. In this sense, the present experiment not only illustrates the theorem's bound and its associated descriptor numerically, but also provides the finite-\(\gamma\) background for the reciprocal-distance burden viewpoint introduced in Section~\ref{sec:reciprocal_distance_burden}.

\subsection{Finite-\texorpdfstring{$\gamma$}{gamma} Validation of Theorem~\ref{thm:large_noise}}
\label{sec:sim_large_gamma}

\begin{figure*}[t]
    \centering
    
    \begin{minipage}[c]{0.73\textwidth}
        \centering
        \includegraphics[width=\linewidth]{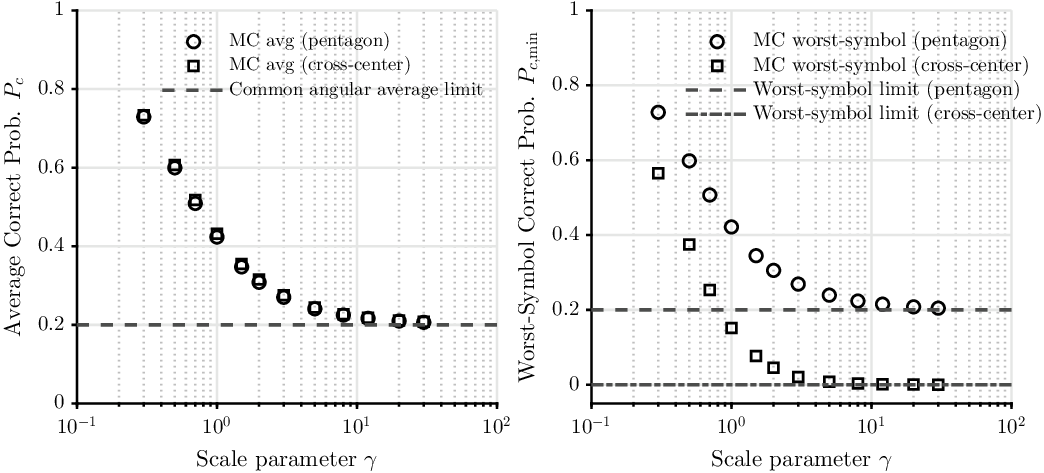}
    \end{minipage}
    \hfill
    \begin{minipage}[c]{0.24\textwidth}
        \centering
        
        \begin{subfigure}{\linewidth}
            \centering
            \begin{tikzpicture}[scale=1.0, >=stealth]
                \draw[gray!70, thin, ->] (-1.35,0) -- (1.35,0) node[right, text=black] {\scriptsize $x$};
                \draw[gray!70, thin, ->] (0,-1.35) -- (0,1.35) node[above, text=black] {\scriptsize $y$};
                
                \draw[dashed, thick, black]
                    (0:1) -- (72:1) -- (144:1) -- (216:1) -- (288:1) -- cycle;
                    
                \foreach \a in {0,72,144,216,288} {
                    \fill (\a:1) circle (1.5pt);
                }
            \end{tikzpicture}
            \caption{Regular pentagon}
            \label{fig:design_hull_pent}
        \end{subfigure}
        
        \vspace{1em} 
        
        \begin{subfigure}{\linewidth}
            \centering
            \begin{tikzpicture}[scale=1.0, >=stealth]
                \pgfmathsetmacro{\r}{sqrt(1.25)}
                
                \draw[gray!70, thin, ->] (-1.35,0) -- (1.35,0) node[right, text=black] {\scriptsize $x$};
                \draw[gray!70, thin, ->] (0,-1.35) -- (0,1.35) node[above, text=black] {\scriptsize $y$};
                
                \draw[dashed, thick, black]
                    (\r,0) -- (0,\r) -- (-\r,0) -- (0,-\r) -- cycle;
                    
                \fill (0,0) circle (1.5pt);
                \fill (\r,0) circle (1.5pt);
                \fill (-\r,0) circle (1.5pt);
                \fill (0,\r) circle (1.5pt);
                \fill (0,-\r) circle (1.5pt);
                
                \draw[->, gray!70, thin] (0.40,0.38)
                    node[right, text=black, font=\scriptsize] {\scriptsize Interior}
                    -- (0.06,0.04);
            \end{tikzpicture}
            \caption{Cross-with-center}
            \label{fig:design_hull_cross}
        \end{subfigure}
        
    \end{minipage}
    
    \vspace{0.5em}
    
    \caption{Descriptor-guided comparison for the hull-only principle under a common average power budget \(P=1\). The regular pentagon \(\mathcal C_{\mathrm{pent}}\) is compared with the cross-with-center design \(\mathcal C_{\mathrm{cross}}\). \textbf{Left:} average correct-decision probability versus \(\gamma\). \textbf{Right:} worst-symbol correct-decision probability versus \(\gamma\). Both constellations approach the same average large-noise limit \(1/5\), but only the pentagon retains a nonzero worst-symbol limit; the cross-with-center design collapses toward zero because of its interior symbol. \textbf{Insets (a) and (b)} show the corresponding constellation geometries.}
    \label{fig:design_hull}
\end{figure*}

We next validate the large-noise prediction in Theorem~\ref{thm:large_noise} using the same asymmetric four-point constellation from Appendix~\ref{app:four_point},
\begin{equation}
    \mathcal{X}=\{(0,0),(1,0),(0,1),(0,-1)\}.
    \label{eq:sim2_constellation}
\end{equation}
For each
\begin{equation}
    \gamma \in \{0.2, 0.3, 0.5, 0.7, 1, 1.5, 2, 3, 5, 8, 12, 20, 30\},
    \label{eq:sim2_gamma_grid}
\end{equation}
we estimate the correct-decision probabilities by Monte Carlo simulation and compare them with the angular limits predicted by Theorem~\ref{thm:large_noise}.

For this constellation, Appendix~\ref{app:four_point} gives the conditional large-noise limits
\begin{equation}
\begin{aligned}
    \lim_{\gamma\to\infty} P_c(P_1;\gamma) &= 0,\\
    \lim_{\gamma\to\infty} P_c(P_2;\gamma) &= \frac{1}{4},\\
    \lim_{\gamma\to\infty} P_c(P_3;\gamma) &= \lim_{\gamma\to\infty} P_c(P_4;\gamma) = \frac{3}{8},
\end{aligned}
\label{eq:sim2_cond_limits}
\end{equation}
and hence the average correct-decision probability satisfies
\begin{equation}
    \lim_{\gamma\to\infty} P_c(\gamma)=\frac{1}{4}.
    \label{eq:sim2_avg_limit}
\end{equation}

Fig.~\ref{fig:baseline_validations}(c)--(d) summarizes the corresponding finite-\(\gamma\) trend toward the large-noise angular limit. In Fig.~\ref{fig:baseline_validations}(c), the average correct-decision probability decreases toward \(1/4\), in agreement with Theorem~\ref{thm:large_noise}. Fig.~\ref{fig:baseline_validations}(d) shows that the conditional correct-decision probabilities approach their distinct angular limits predicted by Appendix~\ref{app:four_point}.

Most notably, the symbol \(P_1=(0,0)\) becomes asymptotically unreliable, with its correct-decision probability tending to zero. This provides direct finite-\(\gamma\) evidence for the geometric-collapse phenomenon in Corollary~\ref{cor:collapse}. By contrast, \(P_2\) approaches the nonzero limit \(1/4\), while \(P_3\) and \(P_4\) approach the common limit \(3/8\), consistent with the symmetry of the constellation.

Together with the previous subsection, this experiment confirms that the large-noise angular descriptor is already numerically meaningful before the asymptotic limit is fully reached. It therefore provides the numerical backdrop for the hull-only principle and collapse-avoidance discussion in Section~\ref{sec:hull_only_principle}.

\subsection{Descriptor-Guided Design Comparison}
\label{sec:design_comparison}

\begin{figure*}[t]
    \centering
    
    \begin{minipage}[c]{0.73\textwidth}
        \centering
        \includegraphics[width=\linewidth]{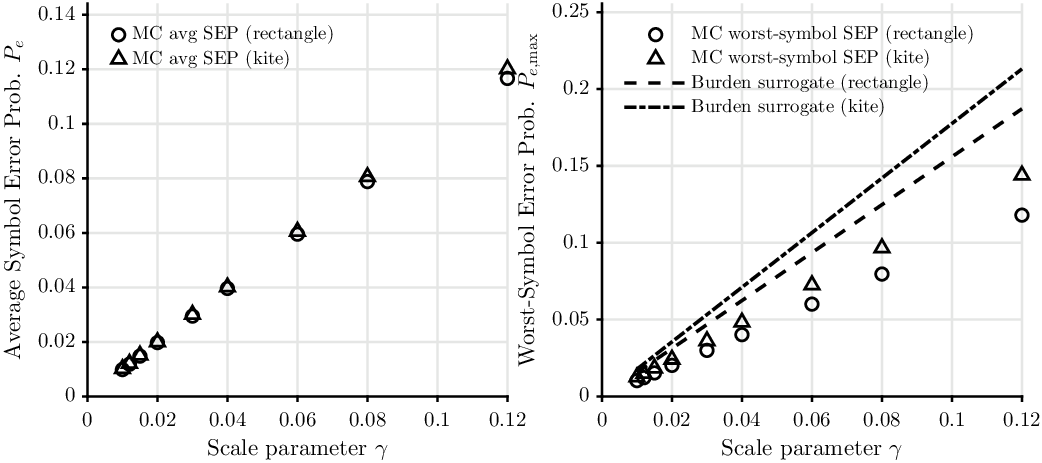}
    \end{minipage}
    \hfill
    \begin{minipage}[c]{0.24\textwidth}
        \centering

        
        \begin{subfigure}{\linewidth}
            \centering
            \begin{tikzpicture}[scale=1.15, >=stealth]
                \pgfmathsetmacro{\hy}{sqrt(3/8)}
                
                \draw[gray!70, thin, ->] (-1.2,0) -- (1.2,0) node[right, text=black] {\scriptsize $x$};
                \draw[gray!70, thin, ->] (0,-1.2) -- (0,1.2) node[above, text=black] {\scriptsize $y$};
                
                \draw[dashed, thick, black]
                    (0.5,\hy) -- (-0.5,\hy) -- (-0.5,-\hy) -- (0.5,-\hy) -- cycle;
                    
                \fill (0.5,\hy) circle (1.5pt);
                \fill (-0.5,\hy) circle (1.5pt);
                \fill (-0.5,-\hy) circle (1.5pt);
                \fill (0.5,-\hy) circle (1.5pt);
            \end{tikzpicture}
            \caption{Rectangle}
            \label{fig:design_burden_rect}
        \end{subfigure}
        
        \vspace{1em}
        
        \begin{subfigure}{\linewidth}
            \centering
            \begin{tikzpicture}[scale=1.15, >=stealth]
                \draw[gray!70, thin, ->] (-1.2,0) -- (1.2,0) node[right, text=black] {\scriptsize $x$};
                \draw[gray!70, thin, ->] (0,-1.2) -- (0,1.2) node[above, text=black] {\scriptsize $y$};
                
                \draw[dashed, thick, black]
                    (0.5,0) -- (0,1) -- (-0.5,0) -- (0,-1) -- cycle;
                    
                \fill (0,1) circle (1.5pt);
                \fill (0.5,0) circle (1.5pt);
                \fill (0,-1) circle (1.5pt);
                \fill (-0.5,0) circle (1.5pt);
            \end{tikzpicture}
            \caption{Kite}
            \label{fig:design_burden_kite}
        \end{subfigure}
        
    \end{minipage}
    
    \vspace{0.5em}
    
    \caption{Descriptor-guided comparison for the reciprocal-distance burden under a common average power budget \(\mathbb E[\|X\|^2]=5/8\). The rectangle \(\mathcal C_{\mathrm{rect}}\) and the kite \(\mathcal C_{\mathrm{kite}}\) are matched in minimum distance \((d_{\min}=1)\) but have different worst-symbol burdens \(B_{\max}\). \textbf{Left:} average symbol error probability versus \(\gamma\). The rectangle remains slightly more reliable throughout the tested range, although the separation is modest at the average level. \textbf{Right:} worst-symbol symbol error probability versus \(\gamma\), together with the linear burden surrogates \((2\gamma/\pi)B_{\max}\). The kite exhibits a consistently larger worst-symbol error probability, and the two surrogates preserve the correct ordering suggested by \(B_{\max}\). \textbf{Insets (a) and (b)} show the corresponding constellation geometries. These results show that the reciprocal-distance burden retains useful finite-\(\gamma\) screening value beyond \(d_{\min}\)-based reasoning.}
    \label{fig:design_burden}
\end{figure*}

The baseline validations in Sections~\ref{sec:sim_small_gamma} and~\ref{sec:sim_large_gamma} confirm that the two asymptotic descriptors remain numerically visible at finite noise scales. We now turn to a more design-oriented question: whether these descriptors also provide useful screening value when candidate planar constellations are compared under a common average power budget. To keep the present extension minimal, we consider two targeted comparison experiments, one for the hull-only principle in Section~\ref{sec:hull_only_principle} and one for the reciprocal-distance burden in Section~\ref{sec:reciprocal_distance_burden}.

\paragraph*{1) Hull-only versus interior-point designs}
For the first experiment, we compare a hull-only five-point design with a five-point design containing an interior point. Under the common average power budget \(P=1\), let
\begin{equation}
    \mathcal C_{\mathrm{pent}} := \left\{ \left(\cos\frac{2\pi k}{5},\,\sin\frac{2\pi k}{5}\right) :\ k=0,1,2,3,4 \right\}
    \label{eq:pentagon_constellation}
\end{equation}
denote the regular pentagon, and let
\begin{equation}
\begin{aligned}
    \mathcal C_{\mathrm{cross}} := \Big\{ &(0,0),\, \left(\sqrt{5/4},0\right),\, \left(-\sqrt{5/4},0\right), \\
    & \left(0,\sqrt{5/4}\right),\, \left(0,-\sqrt{5/4}\right) \Big\}
\end{aligned}
\label{eq:cross_center_constellation}
\end{equation}
denote a cross-with-center design. The former is hull-only, whereas the latter contains an interior point and therefore satisfies
\(
    A_{\min}(\mathcal C_{\mathrm{cross}})=0.
    \label{eq:amin_cross}
\)
By contrast, all five vertices of \(\mathcal C_{\mathrm{pent}}\) have equal exterior angle \(2\pi/5\), so
\(
    A_{\min}(\mathcal C_{\mathrm{pent}})=\frac15.
    \label{eq:amin_pent}
\)
It is worth noting that these two constellations share the same average large-noise limit:
\begin{equation}
    \lim_{\gamma \to \infty} P_c^{(\mathrm{pent})}(\gamma) = \lim_{\gamma \to \infty} P_c^{(\mathrm{cross})}(\gamma) = \frac{1}{5},
    \label{eq:large_noise_limit2}
\end{equation}
but differ sharply at the worst-symbol level. This makes them a particularly clean pair for testing whether the hull-only principle has practical screening value at finite noise scales.

Accordingly, we report both the average correct-decision probability and the worst-symbol correct-decision probability over the grid
\begin{equation}
    \gamma \in \{0.3, 0.5, 0.7, 1, 1.5, 2, 3, 5, 8, 12, 20, 30\}.
    \label{eq:sim3_gamma_grid}
\end{equation}
The average curve probes the coarse angular behavior of the overall design, whereas the worst-symbol curve directly tests whether interior placement leads to finite-\(\gamma\) fragility consistent with the collapse mechanism identified in Corollary~\ref{cor:collapse}.

Fig.~\ref{fig:design_hull} confirms that the decisive separation appears at the worst-symbol level rather than in the average large-noise behavior. Although the two constellations have nearly indistinguishable average correct-decision probabilities and the same asymptotic average limit \(1/5\), the regular pentagon retains a nonzero worst-symbol limit, whereas the cross-with-center design collapses toward zero. This is precisely the mechanism predicted in Section~\ref{sec:hull_only_principle}, and it shows that \(A_{\min}(\mathcal C)\) already has practical finite-\(\gamma\) screening value.

\paragraph*{2) Matched-\(d_{\min}\) comparison with different reciprocal-distance burden}
For the second experiment, we compare two four-point hull-only constellations that are matched in minimum distance and average power but differ in reciprocal-distance burden. Specifically, let
\begin{equation}
\begin{aligned}
    \mathcal C_{\mathrm{rect}} := \Big\{ & \left(1/2,\sqrt{3/8}\right),\, \left(-1/2,\sqrt{3/8}\right), \\
    & \left(-1/2,-\sqrt{3/8}\right),\, \left(1/2,-\sqrt{3/8}\right) \Big\}
\end{aligned}
\label{eq:rect_constellation}
\end{equation}
and
\begin{equation}
    \mathcal C_{\mathrm{kite}} := \left\{ (0,1),\, \left(1/2,0\right),\, (0,-1),\, \left(-1/2,0\right) \right\}
    \label{eq:kite_constellation}
\end{equation}
Under equiprobable signaling, both constellations satisfy the same average power budget
\begin{equation}
    \mathbb E[\|X\|^2]=\frac{5}{8}
    \label{eq:avg_power}
\end{equation}
and the same minimum distance
\begin{equation}
    d_{\min}(\mathcal C_{\mathrm{rect}}) = d_{\min}(\mathcal C_{\mathrm{kite}}) = 1.
    \label{eq:dmin_match}
\end{equation}
Thus, any separation in the small-noise regime cannot be attributed to \(d_{\min}\) alone.

The distinction instead appears in the worst-symbol reciprocal-distance burden. For the rectangle, symmetry gives
\begin{equation}
    B_{\max}(\mathcal C_{\mathrm{rect}}) = 1+\sqrt{\frac23}+\sqrt{\frac25},
    \label{eq:bmax_rect}
\end{equation}
whereas for the kite,
\begin{equation}
    B_{\max}(\mathcal C_{\mathrm{kite}}) = 1+\frac{4}{\sqrt5}.
    \label{eq:bmax_kite}
\end{equation}
Since the kite's burden is strictly larger,
\begin{equation}
    B_{\max}(\mathcal C_{\mathrm{kite}}) > B_{\max}(\mathcal C_{\mathrm{rect}})
    \label{eq:bmax_comparison}
\end{equation}
Section~\ref{sec:reciprocal_distance_burden} suggests that the kite should be more vulnerable at the worst-symbol level in the high-reliability regime, despite the exact matching in \(d_{\min}\) and average power.

We therefore evaluate both the average symbol error probability and the worst-symbol error probability over the grid
\begin{equation}
    \gamma \in \{0.01, 0.012, 0.015, 0.02, 0.03, 0.04, 0.06, 0.08, 0.12\},
    \label{eq:sim4_gamma_grid}
\end{equation}
which matches the small-noise validation grid in \eqref{eq:sim1_gamma_grid}. The purpose of this experiment is to test whether \(B_{\max}(\mathcal C)\) already provides nontrivial screening power beyond the Gaussian-style minimum-distance heuristic, especially at the worst-symbol level.

Fig.~\ref{fig:design_burden} shows that the separation is modest at the average level but clearer at the worst-symbol level. In particular, the kite consistently exhibits a larger worst-symbol error probability than the rectangle, in agreement with \eqref{eq:bmax_comparison}. The linear surrogates \((2\gamma/\pi)B_{\max}\) are not intended as exact finite-\(\gamma\) predictors, but they remain above the Monte Carlo curves and preserve the correct ordering, indicating that \(B_{\max}(\mathcal C)\) provides useful screening information beyond \(d_{\min}\)-based reasoning.

\subsection{Summary of Design-Oriented Numerical Findings}
\label{sec:design_summary}

Overall, the numerical results support the use of the asymptotic quantities developed in Sections~\ref{sec:small_noise} and~\ref{sec:large_noise} as finite-\(\gamma\) design descriptors. The baseline experiments verify that both the distance-spectrum upper-bound descriptor and the angular descriptor are visible before the limiting regimes are fully reached, while the targeted comparisons show that worst-symbol collapse risk and reciprocal-distance burden can separate candidate geometries beyond average-performance and \(d_{\min}\)-based reasoning. Thus, \((A_{\min},B_{\max})\) provides a lightweight but meaningful screening pair for planar isotropic Cauchy constellations under common power budgets.

\section{Discussion: Limitations and Future Directions}
\label{sec:limitations_future}

Several limitations of the present framework should be kept in mind. First, the noise model is restricted to the isotropic multivariate Cauchy law. This excludes both componentwise independent Cauchy coordinates and more general anisotropic or \(\alpha\)-stable heavy-tailed models, whose induced decision geometry and reliability descriptors may differ substantially. Second, the detection setting is deliberately one-shot and uncoded, with maximum-likelihood decoding induced by Euclidean Voronoi cells. The present analysis therefore isolates symbol-level reliability mechanisms, but does not address coded systems, soft information, sequence detection, memory, or iterative receivers. Third, the main results in Sections~\ref{sec:small_noise} and~\ref{sec:large_noise} are asymptotic. Although Section~\ref{sec:numerical_results} shows that the associated descriptors are already visible at finite noise scales, the theory does not yet provide a full nonasymptotic characterization across the intermediate-\(\gamma\) regime.

The design scope is likewise intentionally modest. Section~\ref{sec:design_implications} is restricted to two-dimensional, equiprobable, finite constellations under a common average power budget, and the quantities \(A_{\min}(\mathcal C)\) and \(B_{\max}(\mathcal C)\) are used only as descriptor-level screening tools. Thus, the present paper does not claim a full constellation synthesis theory, nor does it establish optimality guarantees for any design procedure.

These limitations suggest several directions for future work. On the analytical side, an important next step is to develop finite-\(\gamma\) bridges between the reciprocal-distance and angular regimes. On the modeling side, it would be valuable to extend the framework beyond the isotropic multivariate Cauchy setting to anisotropic laws, general \(\alpha\)-stable models, and mismatched-noise scenarios. On the design side, the descriptor-guided comparisons in Sections~\ref{sec:design_implications} and~\ref{sec:numerical_results} suggest the feasibility of a more complete Cauchy-aware constellation synthesis framework, potentially including numerical optimization, robustness criteria, and higher-dimensional constructions. It is also natural to ask whether the present descriptors can inform coded modulation, soft metrics, or receiver architectures in heavy-tailed communication systems (cf. \cite{brown2000nonparametric,tsihrintzis2002performance,tepedelenlioglu2005diversity}).

\section{Conclusion}
\label{sec:conclusion}

This paper developed a geometric asymptotic perspective on one-shot detection under isotropic multivariate Cauchy noise. Although isotropic Cauchy noise induces the same Euclidean Voronoi decision regions as Gaussian noise, the corresponding reliability laws differ fundamentally because the probability mass assigned to those regions is redistributed by the heavy-tailed noise geometry. 
In the small-noise regime, we established a reciprocal distance-spectrum upper bound showing that the associated error-probability descriptor retains sensitivity to the full constellation geometry. 
In the large-noise regime, we proved that the correct-decision probability converges to an angular limit determined solely by the Voronoi recession cone.
Together, these results formalize a regime-dependent transition from bound-based distance descriptors to angle-based reliability descriptors.

For planar constellations, the descriptor-guided comparisons in Sections~\ref{sec:design_implications} and~\ref{sec:numerical_results} further show that these asymptotic quantities already possess meaningful finite-\(\gamma\) screening value: \(A_{\min}(\mathcal C)\) captures collapse avoidance at the worst-symbol level, while \(B_{\max}(\mathcal C)\) distinguishes candidate geometries beyond the classical minimum-distance heuristic. The worked examples in Appendices~\ref{app:four_point} and~\ref{app:4qam}, together with the finite-\(\gamma\) numerical results in Section~\ref{sec:numerical_results}, illustrate how this transition can lead either to highly nonuniform symbol reliability, including complete geometric collapse, or to symmetric limiting behavior in standard constellations such as \(4\)QAM.

\appendices

\begin{figure*}[ht]
    \centering
    
    \begin{subfigure}[b]{0.53\linewidth}
        \centering
        \resizebox{\linewidth}{!}{%
            \begin{tikzpicture}[scale=2.0, >=Stealth, font=\small]
            \definecolor{cblue}{HTML}{1f77b4}
            \definecolor{cred}{HTML}{d62728}
            \definecolor{cgreen}{HTML}{2ca02c}
            \definecolor{cyellow}{HTML}{ff7f0e} 
            
            \tikzset{
                axis style/.style={->, gray!60, shorten >=-2pt}, 
                point style/.style={circle, fill, inner sep=1.5pt}
            }
            
            \begin{scope}[shift={(0,0)}]
                \node[font=\bfseries, anchor=south] at (0, 1.9) {Voronoi Cells};
                
                \def\Lext{1.6} 
                \def\AxExt{1.8}
            
                \fill[cblue!20] (-\Lext,-0.5) rectangle (0.5,0.5); 
                \fill[cred!20] (0.5,-0.5) -- (\Lext,-\Lext) -- (\Lext,\Lext) -- (0.5,0.5) -- cycle; 
                \fill[cgreen!20] (-\Lext,0.5) -- (0.5,0.5) -- (\Lext,\Lext) -- (-\Lext,\Lext) -- cycle; 
                \fill[cyellow!20] (-\Lext,-0.5) -- (0.5,-0.5) -- (\Lext,-\Lext) -- (-\Lext,-\Lext) -- cycle; 
            
                \draw[axis style] (-\AxExt,0) -- (\AxExt,0) node[right=3pt, gray] {$x$};
                \draw[axis style] (0,-\AxExt) -- (0,\AxExt) node[right=3pt, gray] {$y$};
            
                \draw[thick, cblue!80!black] (0.5,0.5) -- (-\Lext,0.5);
                \draw[thick, cblue!80!black] (0.5,-0.5) -- (-\Lext,-0.5);
                \draw[thick, cblue!80!black] (0.5,0.5) -- (0.5,-0.5);
                \draw[thick, dashed, cblue!80!black] (0.5,0.5) -- (\Lext,\Lext);
                \draw[thick, dashed, cblue!80!black] (0.5,-0.5) -- (\Lext,-\Lext);
            
                \node[point style, cblue!80!black] at (0,0) {}; \node[cblue!80!black, below left] at (0,0) {$P_1$};
                \node[point style, cred!80!black] at (1,0) {};  \node[cred!80!black, right] at (1,0) {$P_2$};
                \node[point style, cgreen!80!black] at (0,1) {}; \node[cgreen!80!black, above left] at (0,1) {$P_3$};
                \node[point style, cyellow!80!black] at (0,-1) {}; \node[cyellow!80!black, below left] at (0,-1) {$P_4$};
            
                \node[cblue!80!black] at (-0.8,0) {$V(P_1)$};
            \end{scope}
            
            \begin{scope}[shift={(4.5,0)}]
                \node[font=\bfseries, anchor=south] at (0, 1.9) {Voronoi Cones / Angular Patches};
                
                \def\R{1.5}
                \def\AxExt{1.8}
                
                \fill[cred!20] (0,0) -- (-45:\R) arc (-45:45:\R) -- cycle;
                \fill[cgreen!20] (0,0) -- (45:\R) arc (45:180:\R) -- cycle;
                \fill[cyellow!20] (0,0) -- (180:\R) arc (180:315:\R) -- cycle;
            
                \draw[thin, gray!40] (0,0) -- (45:\R);
                \draw[thin, gray!40] (0,0) -- (-45:\R);
                \draw[thin, gray!40] (0,0) -- (180:\R);
            
                \draw[line width=2pt, cblue, ->] (0,0) -- (180:\R+0.25) node[left=2pt] {$\mathsf{K}(P_1)$};
                \node[cblue, scale=0.75, fill=white, inner sep=1pt, anchor=south] at (-0.8, 0.05) {Measure = 0};
            
                \draw[dashed, gray!30] (0,0) circle (\R);
                
                \draw[axis style] (-\AxExt,0) -- (\AxExt,0) node[right=3pt, gray] {$u_x$};
                \draw[axis style] (0,-\AxExt) -- (0,\AxExt) node[right=3pt, gray] {$u_y$};
                
                \node[cred!80!black] at (0:0.9) {$90^\circ$};
                \node[cgreen!80!black] at (112.5:0.9) {$135^\circ$};
                \node[cyellow!80!black] at (247.5:0.9) {$135^\circ$};
            \end{scope}
            \end{tikzpicture}
        }
        \caption{Geometric properties}
        \label{fig:voronoi_geometry}
    \end{subfigure}
    \hfill
    \begin{subfigure}[b]{0.45\linewidth}
    \centering
    \renewcommand{\arraystretch}{1.08}
    \setlength{\tabcolsep}{3pt}
    \footnotesize
    \resizebox{\linewidth}{!}{%
    \begin{tabular}{lccc}
    \toprule
    Constellation Point & AWGN & Isotropic Cauchy & Isotropic Cauchy \\ 
    $x$ & (small $\sigma$) & (small $\gamma$) & ($\gamma\to\infty$, $P_e$) \\ \midrule
    $P_1=(0,0)$ & $\approx 3\,Q\!\left(\frac{1}{2\sigma}\right)$ & $\le \frac{6}{\pi}\gamma+O(\gamma^3)$ & $1$ \\[0.4em]
    $P_2=(1,0)$ & $\approx Q\!\left(\frac{1}{2\sigma}\right)$ & $\le \frac{2}{\pi}(1+\sqrt2)\gamma+O(\gamma^3)$ & $3/4$ \\[0.4em]
    $P_3=(0,1)$ & $\approx Q\!\left(\frac{1}{2\sigma}\right)$ & $\le \frac{1}{\pi}(3+\sqrt2)\gamma+O(\gamma^3)$ & $5/8$ \\[0.4em]
    $P_4=(0,-1)$ & $\approx Q\!\left(\frac{1}{2\sigma}\right)$ & $\le \frac{1}{\pi}(3+\sqrt2)\gamma+O(\gamma^3)$ & $5/8$ \\ \bottomrule
    \end{tabular}%
    }
    \vspace{1.2cm} 
    \caption{Quantitative reliability descriptors}
    \label{fig:quantitative_table}
    \end{subfigure}
    
    \caption{
    Geometric and quantitative analysis for the illustrative four-point constellation $\mathcal{X}$ (see Appendix~\ref{app:four_point}). \textbf{(a)} \textbf{Left:} Euclidean Voronoi cells. \textbf{Right:} recession cones and their angular patches on $\mathbb S^1$. The point $P_1$ has a degenerate cone of zero angular measure, whereas $P_2, P_3, P_4$ retain nontrivial patches with opening angles $\pi/2$, $3\pi/4$, and $3\pi/4$, respectively. \textbf{(b)} The corresponding theoretical reliability descriptors showcasing the quantitative transition from Gaussian minimum-distance behavior to Cauchy distance-spectrum upper-bound descriptors and angular limits.}
    \label{fig:combined_voronoi_table}
\end{figure*}
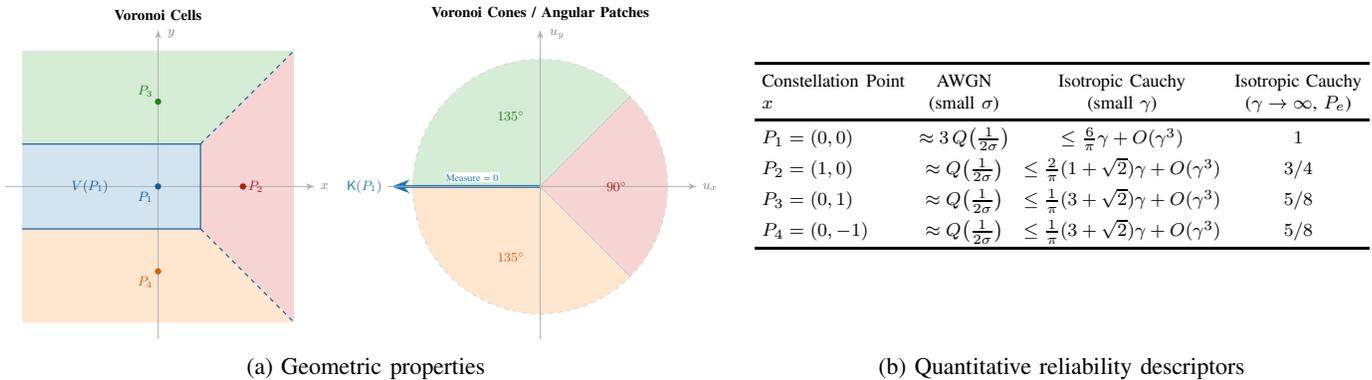

\section{Geometric Collapse: A Four-Point Example}
\label{app:four_point}

To concretely illustrate Theorems~\ref{thm:small_noise} and~\ref{thm:large_noise}, we analyze the four-point constellation $\mathcal X=\{P_1,P_2,P_3,P_4\}\subset\mathbb R^2$ with $P_1=(0,0)$, $P_2=(1,0)$, $P_3=(0,1)$, and $P_4=(0,-1)$. This example is deliberately minimal: all points share the same nearest-neighbor distance $d_{\min}=1$, yet their large-noise behavior under isotropic Cauchy noise differs substantially because their Voronoi recession cones have different angular measures.

The AWGN column is included only as a qualitative comparison: under Gaussian noise, the nearest-neighbor structure largely governs the high-reliability regime, whereas under the isotropic Cauchy model the small-noise upper bound depends on the full reciprocal distance spectrum and the large-noise limit depends only on angular measures.

\subsection{Small-Noise Limit: Explicit Distance-Spectrum Evaluation}

The pairwise distances are $d_{12}=d_{13}=d_{14}=1$, $d_{23}=d_{24}=\sqrt2$, and $d_{34}=2$. By Theorem~\ref{thm:small_noise}, for each transmitted symbol $P_i$,
\begin{equation}
    P_e(P_i;\gamma)\le \frac{2\gamma}{\pi}\sum_{j\neq i}\frac{1}{d_{ij}}+O(\gamma^3),\qquad \gamma\to 0.
\end{equation}

For $P_1=(0,0)$, all three competitors are at unit distance, so
\begin{equation}
\begin{aligned}
    P_e(P_1;\gamma) &\le \frac{2\gamma}{\pi}(1+1+1)+O(\gamma^3) \\
    &= \frac{6}{\pi}\gamma+O(\gamma^3).
\end{aligned}
\end{equation}

For $P_2=(1,0)$, the distances to $P_1,P_3,P_4$ are $1,\sqrt2,\sqrt2$, respectively, hence
\begin{equation}
\begin{aligned}
    P_e(P_2;\gamma) &\le \frac{2\gamma}{\pi}\left(1+\frac{1}{\sqrt2}+\frac{1}{\sqrt2}\right)+O(\gamma^3) \\
    &= \frac{2}{\pi}(1+\sqrt2)\gamma+O(\gamma^3).
\end{aligned}
\end{equation}

For $P_3=(0,1)$, the distances to $P_1,P_2,P_4$ are $1,\sqrt2,2$, respectively, so
\begin{equation}
\begin{aligned}
    P_e(P_3;\gamma) &\le \frac{2\gamma}{\pi}\left(1+\frac{1}{\sqrt2}+\frac{1}{2}\right)+O(\gamma^3) \\
    &= \frac{1}{\pi}(3+\sqrt2)\gamma+O(\gamma^3).
\end{aligned}
\end{equation}

By symmetry, the same expression holds for $P_4=(0,-1)$:
\begin{equation}
    P_e(P_4;\gamma)\le \frac{1}{\pi}(3+\sqrt2)\gamma+O(\gamma^3).
\end{equation}

These calculations reproduce exactly the isotropic Cauchy small-$\gamma$ column in Fig.~\ref{fig:quantitative_table}. In particular, although all four points share the same minimum distance, their reciprocal distance sums differ, so Theorem~\ref{thm:small_noise} already distinguishes their bound-based small-noise descriptors at leading order.

\subsection{Large-Noise Limit: Explicit Angular-Patch Evaluation}

We now verify Theorem~\ref{thm:large_noise} by computing the recession cone
\begin{equation}
    \mathsf K(x) = \bigcap_{c\in\mathcal X\setminus\{x\}} \{u\in\mathbb R^2:\ u^{\mathsf T}(c-x)\le 0\}
\end{equation}
for each constellation point. In two dimensions, Theorem~\ref{thm:large_noise} reduces to
\begin{equation}
    \lim_{\gamma\to\infty}P_c(x;\gamma) = \frac{\text{arc length of }A_x}{2\pi}, \quad A_x = \mathsf K(x)\cap\mathbb S^1.
    \label{eq:appendix_2d_limit}
\end{equation}

\paragraph*{1) Point $P_1=(0,0)$}
The competitor offsets are $P_2-P_1=(1,0)$, $P_3-P_1=(0,1)$, and $P_4-P_1=(0,-1)$. Hence the cone constraints are $u_x\le 0$, $u_y\le 0$, and $-u_y\le 0$, equivalently $u_x\le 0$ and $u_y=0$. Thus $\mathsf K(P_1)$ is the negative $x$-axis ray, and the angular patch $A_{P_1}$ consists of a single point on $\mathbb S^1$. Its arc length is therefore zero, so
\begin{equation}
    \lim_{\gamma\to\infty}P_c(P_1;\gamma) = 0, \quad \lim_{\gamma\to\infty}P_e(P_1;\gamma) = 1.
\end{equation}
This is the geometric-collapse phenomenon predicted by Corollary~\ref{cor:collapse}.

\paragraph*{2) Point $P_2=(1,0)$}
The competitor offsets are $P_1-P_2=(-1,0)$, $P_3-P_2=(-1,1)$, and $P_4-P_2=(-1,-1)$. Hence $-u_x\le 0$, $-u_x+u_y\le 0$, and $-u_x-u_y\le 0$, which is equivalent to $u_x\ge 0$ and $|u_y|\le u_x$. Therefore $\mathsf K(P_2)$ is a wedge centered on the positive $x$-axis with opening angle $\pi/2$, so its spherical patch has arc length $\pi/2$. By \eqref{eq:appendix_2d_limit},
\begin{equation}
    \lim_{\gamma\to\infty}P_c(P_2;\gamma) = \frac{\pi/2}{2\pi} = \frac{1}{4}, \quad \lim_{\gamma\to\infty}P_e(P_2;\gamma) = \frac{3}{4}.
\end{equation}

\paragraph*{3) Point $P_3=(0,1)$}
The competitor offsets are $P_1-P_3=(0,-1)$, $P_2-P_3=(1,-1)$, and $P_4-P_3=(0,-2)$. Hence $-u_y\le 0$, $u_x-u_y\le 0$, and $-2u_y\le 0$, which reduces to $u_y\ge 0$ and $u_x\le u_y$. Thus $\mathsf K(P_3)$ is the wedge spanning angles from $\pi/4$ to $\pi$, whose opening angle is $3\pi/4$. Therefore
\begin{equation}
    \lim_{\gamma\to\infty}P_c(P_3;\gamma) = \frac{3\pi/4}{2\pi} = \frac{3}{8}, \quad \lim_{\gamma\to\infty}P_e(P_3;\gamma) = \frac{5}{8}.
\end{equation}

\paragraph*{4) Point $P_4=(0,-1)$}
By symmetry, $P_1-P_4=(0,1)$, $P_2-P_4=(1,1)$, and $P_3-P_4=(0,2)$, so the cone constraints are $u_y\le 0$ and $u_x+u_y\le 0$, equivalently $u_y\le 0$ and $u_x\le -u_y$. Hence $\mathsf K(P_4)$ is the reflected wedge spanning angles from $\pi$ to $7\pi/4$, again with opening angle $3\pi/4$. Thus
\begin{equation}
    \lim_{\gamma\to\infty}P_c(P_4;\gamma) = \frac{3\pi/4}{2\pi} = \frac{3}{8}, \quad \lim_{\gamma\to\infty}P_e(P_4;\gamma) = \frac{5}{8}.
\end{equation}

Under equiprobable signaling, symmetry implies that the average correct-decision probability also satisfies
\begin{equation}
    \lim_{\gamma\to\infty}P_c(\gamma) = \frac{1}{4}\left(0+\frac{1}{4}+\frac{3}{8}+\frac{3}{8}\right) = \frac{1}{4},
    \label{eq:appendix_average_pc}
\end{equation}
and hence
\begin{equation}
    \lim_{\gamma\to\infty}P_e(\gamma) = 1-\frac{1}{4} = \frac{3}{4}.
    \label{eq:appendix_average_pe}
\end{equation}

This example shows explicitly that the small-noise and large-noise regimes are governed by fundamentally different geometric descriptors. In the high-reliability regime, the relevant quantity is the reciprocal distance sum appearing in Theorem~\ref{thm:small_noise}; in the infinite-noise regime, all bounded distance information disappears and only the angular patch measure from Theorem~\ref{thm:large_noise} survives.

\section{A Standard 4QAM Sanity Check}
\label{app:4qam}

As a complementary sanity check, we briefly examine the standard $4$QAM constellation $\mathcal Q=\{Q_1,Q_2,Q_3,Q_4\}\subset\mathbb R^2$ with $Q_1=(1,1)$, $Q_2=(-1,1)$, $Q_3=(-1,-1)$, and $Q_4=(1,-1)$. Unlike the asymmetric four-point example in Appendix~\ref{app:four_point}, this constellation is fully symmetric. As a result, both the small-noise distance descriptor and the large-noise angular descriptor are identical for all four symbols.

\subsection{Small-Noise Regime}

For any fixed point, say $Q_1=(1,1)$, the distances to the other three points are $\|Q_1-Q_2\|=2$, $\|Q_1-Q_4\|=2$, and $\|Q_1-Q_3\|=2\sqrt2$. Hence, by Theorem~\ref{thm:small_noise},
\begin{equation}
\begin{aligned}
    P_e(Q_1;\gamma) &\le \frac{2\gamma}{\pi}\left(\frac12+\frac12+\frac{1}{2\sqrt2}\right)+O(\gamma^3) \\
    &= \frac{1}{\pi}\left(2+\frac{1}{\sqrt2}\right)\gamma + O(\gamma^3), \quad \gamma\to 0.
\end{aligned}
\label{eq:4qam_small_noise}
\end{equation}
By symmetry, the same bound holds for $Q_2,Q_3,Q_4$.
Thus all symbols have the same leading-order bound-based small-noise descriptor.

\subsection{Large-Noise Regime}

We next evaluate the recession cones. For $Q_1=(1,1)$, the competitor offsets are $Q_2-Q_1=(-2,0)$, $Q_3-Q_1=(-2,-2)$, and $Q_4-Q_1=(0,-2)$. Therefore, the cone constraints are $-2u_x\le 0$, $-2u_x-2u_y\le 0$, and $-2u_y\le 0$, which imply $u_x\ge 0$ and $u_y\ge 0$. Hence $\mathsf K(Q_1)$ is the first quadrant, whose angular patch on $\mathbb S^1$ has arc length $\pi/2$. By Theorem~\ref{thm:large_noise},
\begin{equation}
    \lim_{\gamma\to\infty} P_c(Q_1;\gamma) = \frac{\pi/2}{2\pi} = \frac{1}{4}, \quad \lim_{\gamma\to\infty} P_e(Q_1;\gamma) = \frac{3}{4}.
    \label{eq:4qam_large_noise}
\end{equation}
Again by symmetry, the same limit holds for all four symbols:
\begin{equation}
    \lim_{\gamma\to\infty} P_c(Q_i;\gamma) = \frac{1}{4}, \quad \lim_{\gamma\to\infty} P_e(Q_i;\gamma) = \frac{3}{4},
    \label{eq:4qam_all_points}
\end{equation}
for $i=1,2,3,4$. Under equiprobable signaling, the average correct-decision probability therefore satisfies
\begin{equation}
    \lim_{\gamma\to\infty}P_c(\gamma) = \frac{1}{4}, \quad \lim_{\gamma\to\infty}P_e(\gamma) = \frac{3}{4}.
    \label{eq:4qam_average_limit}
\end{equation}

This example serves as a useful baseline: for a highly symmetric constellation such as $4$QAM, both the reciprocal distance sum in Theorem~\ref{thm:small_noise} and the angular patch measure in Theorem~\ref{thm:large_noise} are identical across symbols.
In contrast, the asymmetric example in Appendix~\ref{app:four_point} shows that these descriptors can vary substantially from point to point, leading to nonuniform bound-based small-noise descriptors and nonuniform large-noise reliability limits.

\balance
\bibliographystyle{IEEEtran}
\bibliography{awcn-clean}

\end{document}